\documentclass[runningheads,orivec]{llncs}
\usepackage[T1]{fontenc}
\usepackage{graphicx}
\usepackage{marvosym}
\usepackage{amsmath}

\usepackage{xcolor}   
\usepackage{algorithm}
\usepackage{algpseudocode}
\usepackage{cite}
\usepackage{tabularx}
\usepackage{hyperref}

\begin{document}
\title{pBeeGees: A Prudent Approach to Certificate-Decoupled BFT Consensus
}
%
%

\author{Kaiji Yang\inst{1}\and
Jingjing Zhang\inst{2}\and
Junyao Zheng\inst{1}\and\\
Qiwen Liu\inst{1}\and
Weigang Wu\inst{1}\and 
Jieying Zhou\inst{1}\textsuperscript{(\Letter)} 
}

\authorrunning{K. Yang et al.}
%
\institute{
School of Computer Science and Engineering, Sun Yat-sen University 
\email{\{isszjy,wuweig\}@mail.sysu.edu.cn}\\
\email{\{yangkj23,zhengjy28\}@mail2.sysu.edu.cn}
\and
Guangdong University of Foreign Studies  School of Cyber Security\\
\email{zhangjj43@gdufs.edu.cn}
}

\maketitle              
\begin{abstract}
Pipelined Byzantine Fault Tolerant (BFT) consensus is fundamental to permissioned blockchains. However, many existing protocols are limited by the requirement for view-consecutive quorum certificates (QCs). This constraint impairs performance and creates liveness vulnerabilities under adverse network conditions. Achieving ``certificate decoupling''—committing blocks without this requirement—is therefore a key research goal. While the recent BeeGees algorithm achieves this, our work reveals that it suffers from security and liveness issues. To address this problem, this paper makes two primary contributions. First, we formally define these flaws as the Invalid Block Problem and the Hollow Chain Problem. Second, we propose pBeeGees, a new algorithm that addresses these issues while preserving certificate decoupling with no additional computational overhead. To achieve this, pBeeGees integrates traceback and pre-commit validation to solve the Invalid Block Problem. Further, to mitigate the Hollow Chain Problem, we introduce a prudent validation mechanism, which prevents unverified branches from growing excessively. To summarize, pBeeGees is the first protocol to simultaneously achieve safety, liveness, and certificate decoupling in a pipelined BFT framework. Experiments confirm that our design significantly reduces block commit latency compared to classic algorithms, particularly under frequent stopping faults.

\end{abstract}

\section{Introduction}
Blockchain technology has emerged as a transformative force, enabling mutually distrustful parties to compute over shared data in a decentralized manner\cite{cole2019blockchain, fernandez2019blockchain, ali2021security}. At its core, Byzantine Fault Tolerant (BFT) consensus algorithms\cite{lamport_byzantine_1982} provide the fundamental guarantees of data consistency, low commit latency, and high throughput, making them indispensable for permissioned blockchains and sharding-based public blockchain systems\cite{nakamoto2008bitcoin}. In recent years, chained BFT protocols, exemplified by HotStuff\cite{yin2019hotstuff}, Fast-HotStuff\cite{jalalzai2020fast}, and DiemBFT\cite{Team2021DiemBFTVS}, have gained significant traction. These protocols cleverly leverage pipelining and leader rotation mechanisms to achieve high throughput and responsiveness.

However, a critical challenge inherent in these modern chained BFT protocols is their stringent requirement for consecutive honest leaders to commit operations. Consequently, even simple leader failures, such as crashes, can severely weaken system liveness, leading to significantly increased commit latency or even a complete halt of block commitment. Furthermore, malicious participants can easily exploit this vulnerability by merely delaying responses.

To address this problem, BeeGees\cite{giridharan2023beegees} proposed the concept of \textbf{Certificate Decoupling}, which means block commitment can be done with non consecutive QCs. Despite its pioneering contribution, our in-depth analysis reveals inherent security and liveness flaws within the BeeGees algorithm, which could lead to degraded performance or even protocol failure. Specifically, we identify two critical issues: the        ``Invalid Block Problem'' and the  ``Hollow Chain Problem''. The ``Invalid Block Problem'' arises because BeeGees' validation process does not thoroughly track and prevent implicitly invalid blocks, allowing malicious leaders to hide explicitly invalid blocks within the chain, potentially leading to safety violations where conflicting blocks are committed. A single malicious leader can exploit this by proposing an invalid block and then halting, forcing subsequent leaders to extend a seemingly valid but implicitly corrupted chain. The ``Hollow Chain Problem'' manifests under continuous timeouts, leading to the formation of prolonged chains composed of unvalidated blocks. Such hollow chains significantly increase the computational overhead for validation, potentially causing replicas to time out again before completing validation, thereby severely impacting liveness and throughput. 

Main contributions of this work are summarized in the following: 

\begin{itemize}
    \item We formally identify and thoroughly analyze the ``Invalid Block Problem'' and the ``Hollow Chain Problem'' inherent in BeeGees.
    \item We integrate traceback validation and pre-commit validation mechanisms to robustly address the ``Invalid Block Problem,'' ensuring the validity of all ancestor blocks and preventing the propagation of invalid states. Simultaneously, we introduce a prudent validation mechanism that effectively mitigates the ``Hollow Chain Problem'' by preventing unvalidated chains from growing excessively long, thereby maintaining stable performance and liveness even under adverse conditions. 
    \item Building upon pBeeGees, we further propose pBeeGees-CB, which incorporates a novel Commit Boost mechanism. This mechanism leverages the unique properties of pBeeGees to enable blocks to be committed with just one round of voting in optimistic scenarios where all processes vote, significantly reducing block commit latency. This optimization introduces an increase in communication complexity but offers a substantial reduction in latency and can be flexibly enabled or disabled based on network load. 
    
\end{itemize}

The remainder of this paper is organized as follows. Section 2 introduces the system model and preliminaries. Section 3 provides an overview of BeeGees, while Section 4 discusses its limitations. Sections 5 and 6 present pBeeGees and pBeeGees-CB, respectively. Section 7 offers the correctness proof, and Section 8 details the evaluation. Finally, Section 9 reviews related work.

\section{System Model and Preliminaries}
\label{sec:system_model}

This section outlines the fundamental system model and defines the key terminologies used throughout this paper.

\subsection{System Model}
\label{ssec:system_model_details}
We consider a set of participants consisting of \( n \) nodes, where at most \( f \) nodes are malicious($ n = 3 f +1 )$. These malicious nodes may deviate arbitrarily from the protocol, but they have limited computational power and cannot forge digital signatures or message digests. 
We operate under a partial synchrony model\cite{dwork1988consensus}. After
GST, all transmissions arrive within a known bound $\Delta$ to their destinations.
\subsection{Preliminaries}
\label{ssec:preliminaries_details}

\subsubsection{Quorum Certificate (QC)}
A QC proves that more than $n - f$ replicas have voted for a block in a given view. It  serves as evidence that the block is certified. We use $B^{QC}$ to emphasize that block $B$ has collected QC. And we use $B_{QC}$ to denote the block proven by $B.QC$, which is $B.QC.B$.

\subsubsection{Timeout Certificate (TC)}A TC is formed by collecting at least $n - f$ timeout messages for view $v$ from different processes, proving that view $v$ has timed out and prompting others to move to view $v + 1$.

\subsubsection{First and Second Certified Blocks}
On a single chain with a block B as the leaf node, the certified block with the highest view is called the ``First Certified Block,'' while the certified block with the second-highest view is called the ``Second Certified Block.''

\subsubsection{Equivocation}In each view, a leader is expected to propose only one block. However, a malicious leader might propose multiple different blocks to different processes. This act, known as equivocation, is an attempt to compromise the system's security. 

\subsubsection{Block Structure Definition}
A block B is defined as follows:
$$\langle i, v, B_p, QC, B_{QC}, TC, \textit{tmo\_set} \rangle_{\sigma_i}$$
This structure contains all the necessary information for the block, where $i$ is the proposing process, $v$ is the view the block belongs to, $B_p$ is its parent block, $QC$ is the quorum certificate carried by the block, $B_{QC}$ is the block it certifies, $TC$ is the timeout certificate it carries, and $tmo\_set$ is defined as the $n-f$ timeout messages corresponding to the timeout certificate. $\sigma_i$ represents that the block carries the signature of process $i$.

\section{BeeGees in brief}

This section provides a concise overview of the BeeGees algorithm and highlights its core innovation—Certificate Decoupling.

\subsection{Certificate Decoupling}
\begin{figure}[htbp]
    \centering
    \includegraphics[height=0.3\textheight]{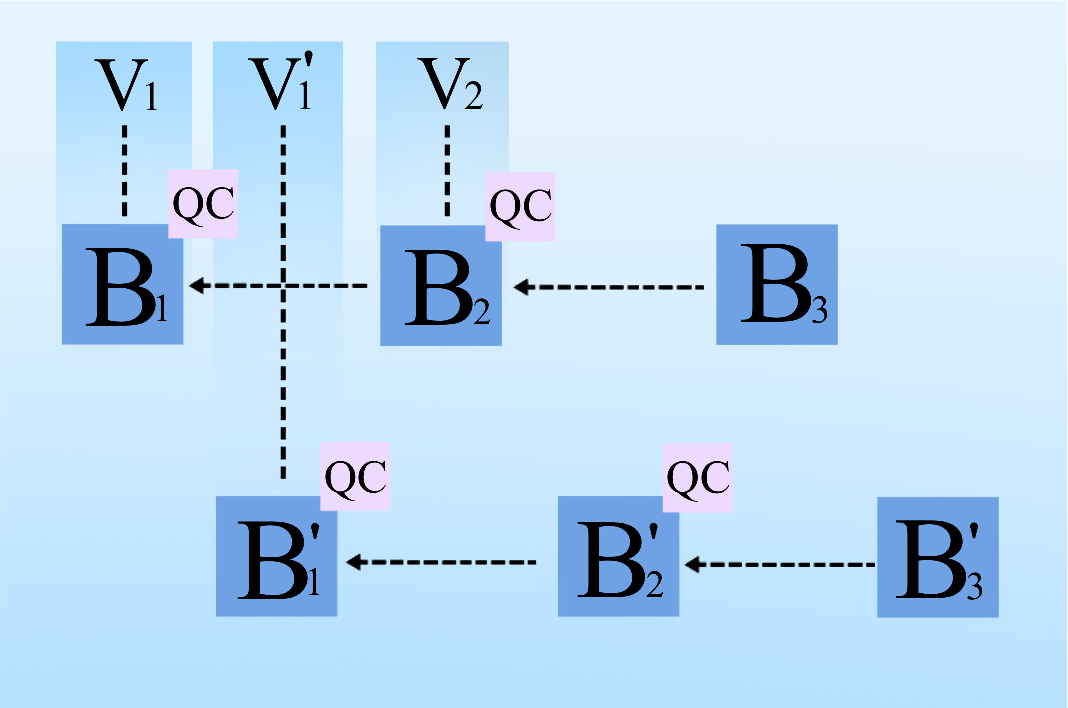}
    \caption{Fast-HotStuff requires that the views associated with certificates are consecutive.}
    \label{fig:fasthotstuff}
\end{figure}

BeeGees's primary contribution lies in achieving the \textit{certificate decoupling} property, which enhances algorithm liveness and reduces block commit latency. To better understand the motivation behind this design, let us take a look at  classic algorithm FastHotStuff. As shown in \autoref{fig:fasthotstuff}, in view $v_3$, a process that receives the proposed block $B_3$ determines whether its second certified block $B^{QC}_{1}$ can be committed. The key condition is that the views of $B^{QC}_{1}$ and the first certified block $B^{QC}_{2}$ must be consecutive, i.e., $v_2 = v_1 + 1$. If this condition is not met, even a certified block cannot be committed immediately, resulting in commit delays.

Understanding why Fast-HotStuff requires consecutive views for block commitment is essential to achieving certificate  
 decoupling.
The reason lies in the pipelined structure, which intertwines the consensus processes of different blocks, leading to the possibility of conflicting QCs. As shown in \autoref{fig:fasthotstuff}, if the view between block \( B_1 \) and \( B_2 \) is not consecutive, a block \( B_1' \) may appear in the meantime that conflicts with \( B_1 \) and also collects a QC. Such conflicting blocks with QCs can keep appearing alternately. If non-consecutive certified blocks are allowed to be committed, it may result in the commitment of conflicting blocks. 

To address the security issues caused by this commitment of non-consecutive block,  BeeGees inspects the blockchain to detect conflicting QCs before committing. If such conflicts are found, BeeGees suspends the commit; otherwise, the block is committed.

\subsection{Operation in Ideal Conditions}

Under ideal conditions, the leader’s proposed blocks consistently receive sufficient votes to become certified (i.e., obtain quorum certificates), and these certified blocks appear in consecutive views. As a result, commitments can proceed continuously.

Considering view v, the leader $l_v$ proposes a block $B_v$ and broadcasts it to all replica nodes. Upon receiving $B_v$, replicas vote for it and send their votes to the leader of the next view, $l_{v+1}$. Once $l_{v+1}$ collects at least $n-f$ votes for $B_v$, it generates a QC for $B_v$. Simultaneously, $l_{v+1}$ proposes a new block $B_{v+1}$, which references $B_v$ as its parent. If $B_{v+1}$ also obtains a QC, then $B_v$ is considered to have passed two consecutive voting phases and is thus deemed \textit{committed}.

\subsection{View Change}

Each process $p$ maintains its \texttt{high\_vote} field, which records its voting information for the block with the highest view, designated as $B_{high}$. Upon a timeout, process $p$ sends a timeout message containing \texttt{high\_vote} to the next leader $l_v$. If $l_v$ collects $n-f$ timeout messages from different processes, it completes the view change and proposes a new block $B$. The crucial aspect here is the selection of $B$'s parent block, $B_{p}$. $l_v$ selects the block $B_{high}$ corresponding to the \texttt{high\_vote} with the highest ranking among the $n-f$ timeout messages as the parent block. As illustrated in \autoref{alg:ranking}, the block with the higher view has a higher rank. If the views are the same, a tie is broken by comparing the view of the QC carried by the blocks.

\begin{algorithm}
\caption{Ranking Rule}
\label{alg:ranking}
\begin{algorithmic}[1]
    \Function{RANK}{$B1, B2$} \Comment{The Block Ranking Rule}
        \State \textbf{return} $(B_1.v > B_2.v) \lor (B_1.v = B_2.v \land B_1.QC.v > B_2.QC.v)$
    \EndFunction
\end{algorithmic}
\end{algorithm}

\subsection{Commit Rule}

Whenever a process encounters a new block $B$, it applies the commit rule to determine if any new block can be committed. Let $B_{1}^{QC}$ and $B_{2}^{QC}$ be the first and second certified blocks on the chain containing $B$, respectively. If their views are consecutive, BeeGees directly commits $B_{2}^{QC}$. If their views are not consecutive, then all blocks between $B_{2}^{QC}$(exclusive) and $B_{1}^{QC}$(inclusive) require sequential block commit validation to mitigate security risks associated with equivocation. The validation rule involves checking if, for each block $B^{*}$ that needs to be checked, its \texttt{tmo\_set} field (which contains $n-f$ timeout messages) includes a block $B^{\prime}$ that has the same view as its parent block $B_{p}^{*}$ and conflicts with $B_{2}^{QC}$. If such a block $B'$ exists, the commit validation fails, and the commit is suspended; otherwise, $B_{2}^{QC}$ can be committed.

\section{Flaws in the BeeGees Protocal}

The practice of using vote information to select a parent block can solve the previously mentioned problem of conflicting blocks, but it also introduces three new problems. One of these is equivocation\cite{jaffe2012price}, which is addressed by BeeGees through its commit rule. However, BeeGees overlooks the second and third problems, which we term the \textbf{Invalid Block Problem} and the \textbf{Hollow Chain Problem}. These two issues pose challenges to the safety and liveness of BeeGees but were not discussed in detail in its paper. The following discussion will provide a detailed explanation and analysis of them.

\subsection{The Invalid Block Problem}

We first define invalid blocks, and then introduce the potential risks associated with the issue of invalid blocks.

An \textbf{Explicitly Invalid Block} is defined as a block that does not comply with the validation rules described in the \autoref{tab:explicitly_valid_block}.

\begin{table}[htbp]
\centering

\caption{Explicitly Valid Block}
\label{tab:explicitly_valid_block}
\begin{tabularx}{\textwidth}{|p{4cm}|X|}
\hline
\textbf{Block Type} & \textbf{Validation Rules} \\
\hline
\textbf{Block generated after a timeout} & 1. The signature, $i$, and $v$ are valid and consistent. \newline 2. The $QC$ is valid and consistent with $B_{QC}$. \newline 3. $B_{QC}$ is an ancestor of $B$. \newline 4. $tmo\_set$ contains timeout messages from $n-f$ different processes and is consistent with the $TC$. \newline 5. $B_p$ is consistent with one of the blocks, $B'$, in $tmo\_set$, and for all other $B''$,  $Rank(B') \geq Rank(B'')$. \\
\hline
\textbf{Block generated after a vote} & 1. The signature, $i$, and $v$ are valid and consistent. \newline 2. The $QC$ is valid and consistent with $B_{QC}$. \newline 3. $B_{QC}$ is the parent of $B$. \\
\hline
\end{tabularx}
\end{table}

A block that passes the validation is called an \textbf{explicitly valid block}, as the validation process mainly checks the information carried by the block itself rather than its ancestors. It should be noted that, according to this definition, even if a block is explicitly valid, it does not guarantee that all its information is valid. For example, the \texttt{tmo\_{set}} of the block may contain other invalid blocks. However, although such erroneous information exists in a block, it will not affect the algorithm's safety. Adopting this definition of an explicitly invalid block helps reduce the algorithm's computational complexity.

An \textbf{Implicitly Invalid Block} is a block that is itself explicitly valid, but one of its ancestors is an explicitly invalid block.

Explicitly invalid blocks and implicitly invalid blocks will be collectively referred to as \textbf{Invalid Blocks}, which undermine algorithm safety. The safety proof of BeeGees assumes that all relevant blocks are valid and does not explore the specific block validation process \cite{giridharan2023beegees}. \autoref{fig:InvalidBlock} illustrates a scenario that compromises safety. Before view 4, the system operates normally, and both $B_3$ and $B_2$ have collected QC, so $B_2$ is committed. However, the leader of view 5 is a malicious process, and the block it proposes, $B_5$, conflicts with $B_2$. Clearly, $B_5$ is an explicitly invalid block and cannot collect a quorum certificate. But at this point, $l_5$ chooses to cease operation, triggering a timeout, and sends a timeout message to $l_6$. $l_6$ is also a malicious process, and it chooses $B_5$ as the parent of its block. Note that at this point, $B_6$ is no longer an explicitly invalid block. If $l_7$ is also a malicious leader, this scenario can continue, and the only explicitly invalid block, $B_5$, becomes hidden deep in the chain. If validation only checks the validity of the current block, then $B_6$ and subsequent blocks are all considered valid and can be committed in subsequent operations, but they all conflict with the committed block $B_2$.

\begin{figure}[htbp]
    \centering
    \includegraphics[width=0.95\textwidth]{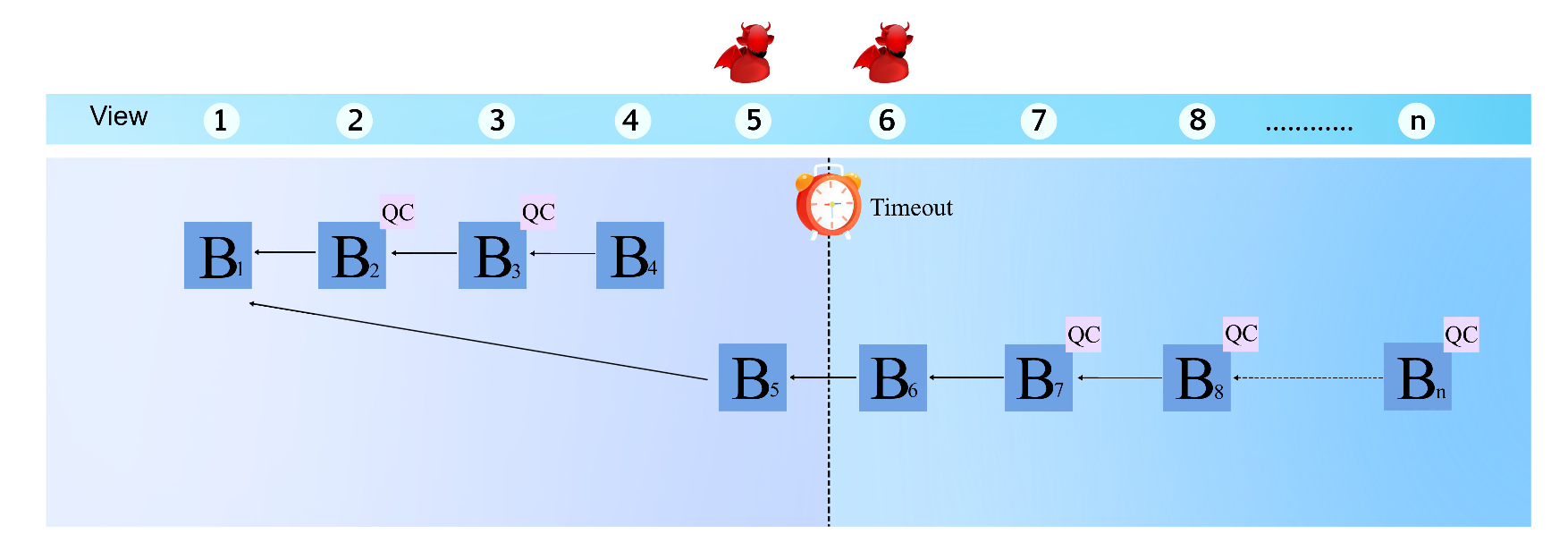}
    \caption{Invalid blocks undermine security.}
    \label{fig:InvalidBlock}
\end{figure}

\subsection{The Hollow Chain Problem}

In the case of consecutive timeouts, a long chain composed of unverified blocks can be generated. Consider the following situation: a block $B_v$ is proposed in every view, but before it can collect a QC, the system times out. $B_v$ is then included in a timeout message passed to the leader of the next view, $l_{v+1}$, and is chosen as the parent of $B_{v+1}$. At this point, neither $B_{v+1}$ nor $B_v$ has collected a vote certificate. This process can continue indefinitely, eventually forming a chain of unverified blocks on the blockchain. We call such a chain a \textbf{Hollow Chain}.

The aforementioned consecutive timeouts are not difficult to trigger. Before GST, since there is no upper bound on message transmission delay between nodes, consecutive timeouts can occur rapidly, and the rate of occurrence is proportional to the performance of the nodes. Even after GST, some view synchronizers that rely on manually setting an initial view and converging in subsequent operations can also lead to the creation of hollow chains.

A hollow chain affects both the algorithm's safety and liveness. Because processes need to perform a complete verification of the hollow chain, it impacts the algorithm's performance. In extreme cases, if the hollow chain is too long, it is even possible for a process to time out before completing the verification, leading to a vicious cycle. Furthermore, the liveness analysis of traditional consensus algorithms assumes that the block verification time can be ignored. In theory, there is no upper limit to the length of a hollow chain, so making the same assumption may no longer be appropriate. This would directly affect the applicability of these liveness theories. It might even become difficult to determine whether the algorithm can complete within a finite time, even after GST.

\section{pBeeGees Algorithm}

\subsection{Design Optimizations}
For ease of understanding, we gradually transition from the BeeGees algorithm to the pBeeGees algorithm through several optimization steps.
\subsubsection{Traceback Validation}

To address the issue in BeeGees where malicious processes can hide illegal blocks, we introduce traceback validation, which requires a process to ensure the validity of block $B$ when voting for or selecting a parent block. If $B$ and its ancestors are generated after a timeout, they must be recursively validated back to the first justified block, $B_{\text{QC}}$. The core of the validation is to ensure that for each block to be verified, $B^{*}$, its parent block, $B^{*}_{p}$, has the highest rank in the \texttt{high\_vote} carried within $B^{*}$'s \texttt{tmo\_set}, and that $B^{*}_{p}$ is itself valid. Leaders must also perform this check when selecting a parent block to prevent the protocol from stalling by building on an invalid chain.

\subsubsection{Pre-Commit Validation}

While traceback validation ensures safety, it may cause redundant re-verification. For example, assume there is a  block $B_3$, which has a first certified block $B_2$ and a second certified block $B_1$. Validating a block $B_2$ in one view and again validating $B_3$ that builds on $B_2$ causes duplication. To optimize, pBeeGees performs pre-commit validation concurrently with block validation. During validation, a flag \texttt{FLAG\_EQVC} is set if any equivocation is detected—i.e., conflicting parent choices in the timeout set. 

Depending on the outcome, replicas vote with one of two types: \texttt{vote\_normal}, which indicates no equivocation was found, and \texttt{vote\_eqvc}, which signals that equivocation was detected. The leader aggregates votes of the same type to form either $QC_{normal}$ or $QC_{eqvc}$, respectively. At this point, when a process receives block $B_3$, if it finds that the type of QC $B_3$ carries is $QC_{normal}$, it directly commits the second-certified block $B_1$. If the QC type is $QC_{eqvc}$, it only performs a validity check on $B_3$ without attempting to commit $B_1$.

\subsubsection{Prudent Validation}

In a prolonged timeout scenario, the chain may grow with unvalidated blocks, reducing liveness. To prevent these ``hollow chains,” we define a \textit{prudence degree} (PD), denoting the maximum number of consecutive timeout-generated blocks without a QC.

Each timeout-generated block increments a counter \texttt{cnt\_tmo}. If \texttt{cnt\_tmo} exceeds $PD$, the block is considered invalid. If it equals $PD$, the vote type is set to \texttt{vote\_prud}, indicating a prudent block. This mechanism forces QCs to form periodically, preventing the chain from growing indefinitely without justification.  Once enough \texttt{vote\_prud} votes are collected, a new block marked with \texttt{QC\_prud} is generated.  
 \texttt{QC\_prud} only proves block validity and cannot be used for commitment or ordering.

\subsection{Protocol Overview}

The pBeeGees protocol consists of three main components: block proposal, validation, and commitment.

\subsubsection{Block Proposal}
The block proposal algorithm is described in \autoref{alg:block_proposal}.
If the current leader process enters view $v$ via $QC_{v-1}$, this means that the leader of previous view successfully proposed block $B_{v-1}$ which later obtained sufficient votes. Specificially, the leader $l_v$ of view $v$ collects votes of the same type from $n-f$ distinct processes in view $v-1$ for block $B_{v-1}$, it generates a corresponding QC and then proposed block $B_v$, which uses $B_{v-1}$ as its parent and includes the QC in its QC field. Afterwards, $l_v$ broadcasts block $B_v$.

If the current leader process enters view $v$ via $TC_{v-1}$, meaning the previous view timed out, then the proposal process is as follows. The leader $l_v$ first collects $n-f$ valid timeout messages for view $v-1$, forms $TC_{v-1}$, and enters view $v$. From the \texttt{high\_vote} fields in the timeout messages, $l_v$ selects the valid block $B_{\text{high}}$ with the highest ranking as the parent block, generates a new block $B_v$, and includes $TC_{v-1}$ in the corresponding field. The QC field of $B_v$ uses the QC field of $B_{\text{high}}$, and the \texttt{cnt\_tmo} field is set as $B_v.\text{cnt\_tmo} = B_{high}.\text{cnt\_tmo} + 1$ to count the hollow chain length.

\begin{algorithm}[htbp]
\caption{Block Proposal}
\label{alg:block_proposal}
\begin{algorithmic}[1]
\Function{PROPOSE\_BY\_QC}{qc, cmd}
    \State curr\_view $\gets$ synchronizer.v
    \State parent $\gets$ qc.block
    \State b $\gets$ NEW\_BLOCK(parent, qc, cmd, curr\_view, proposer\_id)
    \State BROADCAST(b)
\EndFunction

\Function{PROPOSE\_BY\_TC}{tc, tmo\_set, cmd}
    \State curr\_view $\gets$ synchronizer.v
    \State parent $\gets$ GET\_HIGH\_BLOCK(tmo\_set)
    \State qc $\gets$ parent.qc
    \State cnt\_tmo $\gets$ parent.cnt\_tmo + 1
    \State b $\gets$ NEW\_BLOCK(parent, qc, cmd, curr\_view, proposer\_id, cnt\_tmo, tc, tmo\_set)
    \State BROADCAST(b)
\EndFunction
\end{algorithmic}
\end{algorithm}

\subsubsection{Block Validation}
As shown in \autoref{alg:validation}, the general process of block verification is discussed.
Specifically, if a block is produced through the normal voting process, meaning it carries a valid quorum certificate, its validation is relatively simple. It mainly involves checking the validity of the carried QC and ensuring the view number is consecutive. If validation passes, the block is considered valid, and a vote of type \texttt{vote\_normal} can be cast for it.

However, if a block is generated due to a timeout mechanism, the validation process is more complex. In this case, the algorithm must perform a recursive validation on all blocks from the current block back to the most recent block on its chain that holds any type of QC ($B_{QC}$). During this traceback validation, each block undergoes both validity verification and pre-commit validation. The core of the validity verification is to ensure that the parent block chosen by the current block is the highest-ranked among all timeout messages it collected, and that the parent block itself is valid. The algorithm also checks that the length of the ``hollow chain'' does not exceed a preset Prudence Degree threshold. 

While performing validity verification, the algorithm simultaneously executes pre-commit validation to proactively detect any potential equivocation. If a conflicting equivocation is found on the chain, the corresponding vote type is marked as \texttt{eqvc}. Furthermore, if the ``hollow chain'' length of a block reaches the prudence threshold, it is treated as a ``prudent block,'' and the vote type is marked as \texttt{prud}. The conditions for marking a vote as \texttt{vote\_eqvc} and \texttt{vote\_prud} are orthogonal. A vote can therefore be assigned a composite type if a block's validation satisfies both criteria simultaneously. Combining these conditions, pBeeGees designs a total of four vote and corresponding QC types.(See \autoref{tab1}). Finally, the replica process updates its local \texttt{high\_vote} record based on the validation result: if the current block is a prudent block, it records its parent block; otherwise, it records the current block itself. This entire validation flow, through recursive traceback and a multi-type voting mechanism, ensures security while resolving the invalid block and hollow chain problems.
\begin{table}
\centering
\caption{Concise Description of Vote Types.}\label{tab1}
\begin{tabular}{|l|l|l|}
\hline
Vote Type &  Description & Corresponding QC Type\\
\hline
$vote_{normal}$ &  Non-prudent block, validation passed  & $QC_{normal}$\\
$vote_{prud}$ &  Prudent block, validation passed & $QC_{prud}$\\
$vote_{eqvc}$ & Non-prudent block, validation failed  & $QC_{eqvc}$ \\
$vote_{prud,eqvc}$ & Prudent block, validation failed & $QC_{prud,eqvc}$\\

\hline
\end{tabular}
\end{table}

\begin{algorithm}
\caption{Block Validation Algorithm}
\label{alg:validation}

\begin{algorithmic}[1]
\Function{Valid\_Chain}{b}
  \State $b_p \gets b.p$
  \State $b_{qc} \gets b.qc.block$
  \If{b was formed by votes}
      \State \Call{Valid\_QC}{b.qc, b.v}  
      \State \Return \text{NORMAL}
  \EndIf
  \State \Call{Valid\_TC}{b.tc, b.tmo\_set, b.v} 
  \State tmo\_set\_blocks $\gets$ \Call{Get\_Tmo\_Set\_Blocks}{b.tmo\_set}
  \State \Call{Sort\_By\_Rank}{tmo\_set\_blocks}
 
  \For{each $b_i$ \textbf{in} tmo\_set\_blocks}
    \If{$b_p = b_i$}
      \State result $\gets$ \Call{Valid\_Chain}{$b_p$}
      \If{result = \text{INVALID}}
        \State \Return \text{INVALID}
      \ElsIf{result = \text{EQVC}}
        \State FLAG\_EQVC $\gets$ \textbf{true}
      \EndIf
      \State FLAG\_EXIST $\gets$ \textbf{true}
    \ElsIf{\Call{Rank}{$b_i$} > \Call{Rank}{$b_p$}}
      \State \Return \text{INVALID}
   \ElsIf{\Call{Rank}{$b_i$} = \Call{Rank}{$b_p$}}
      \If{FLAG\_EQVC = \textbf{false}}
       \State put $b_i$ into equivoc\_block\_list
      \EndIf
    \Else
      \State \textbf{break}
    \EndIf
  \EndFor
 
  \If{\textbf{not} FLAG\_EXIST}
    \State \Return \text{INVALID}
  \EndIf
 
  \If{FLAG\_EQVC = \textbf{false}}
    \State result $\gets$ \Call{Test\_Equivocation}{equivoc\_block\_list, $b_{qc}$}
  \EndIf
 \State \Return result
\EndFunction
\end{algorithmic}
\end{algorithm}

\subsubsection{Commit Rule}

Since the block has already undergone pre-commit validation during its validity verification, the block commit validation process in pBeeGees is very simple. A block certified by a non-\textit{prud} type QC is called $B_{QC_{\text{!prud}}}$ (i.e., $B_{QC_{\text{normal}}}$ or $B_{QC_{\text{eqvc}}}$). Whenever a process $p$ receives a new valid block, or when $p$, as the leader, generates a new block, $p$ performs block commit validation. Let the target block be $B
$, and let the closest and second closest blocks that obtains non-\textit{prud} type QC be $B_{QC_{\text{!prud}}1}$ and $B_{QC_{\text{!prud}}2}$. If the type tag of $B_{QC_{\text{!prud}}1}.QC$ does not contain \textit{eqvc}, it means that there is no equivocal proposal between $B_{QC_{\text{!prud}}1}$ and $B_{QC_{\text{!prud}}2}$, and $B_{QC_{\text{!prud}}2}$ can be committed; otherwise, the block is not committed. For details, refer to \autoref{alg:commit}

\begin{algorithm}
\caption{Block commit Algorithm}
\label{alg:commit}
\begin{algorithmic}[1]
\Function{COMMIT\_RULE}{b}
    \State $b_{qc_{!prud}1}, qc\_type \gets \text{GET\_NONPRUD\_QC\_BLOCK}(b)$ 
    \State $b_{qc_{!prud}2}, \_ \gets \text{GET\_NONPRUD\_QC\_BLOCK}(b_{qc_{!prud}1})$ 
    \If{$qc\_type == QC_{normal}$} 
        \State \textbf{return} $b_{qc_{!prad}2}$ 
    \EndIf
    \State \textbf{return} NULL 
\EndFunction
\\
\Function{GET\_NONPRUD\_QC\_BLOCK}{b}
    \State $qc \gets b.qc$
    \State $vote\_type \gets qc.vote\_type$
    \If{$vote\_type == vote_{prud}$} 
        \State $b_{prud} \gets b.qc.b$
        \State \textbf{return} $\text{GET\_NONPRUD\_QC\_BLOCK}(b_{prud})$ 
    \EndIf
    \State \textbf{return} $b.qc.b, vote\_type$
\EndFunction
\end{algorithmic}
\end{algorithm}

\subsection{Communication Complexity}

The pBeeGees algorithm has the same inter-process communication pattern as BeeGees, and thus it has the same communication complexity. 

In normal operation without timeouts, the block broadcasted by the leader has a size of $O(1)$ and is sent to $O(N)$ replica nodes. Meanwhile, each replica's vote message is also of size $O(1)$. Therefore, the message complexity is $O(N)$.
When a timeout occurs, the proposed block contains a \texttt{high\_vote} message of size $O(N)$ and is broadcast to all replicas, resulting in an overall message complexity of $O(N^2)$.

\section{pBeeGees-CB Algorithm}
This section presents the pBeeGees-CB algorithm. Based on the pBeeGees algorithm proposed in the previous section, pBeeGees-CB incorporates a Commit Boost mechanism to further reduce the block commit latency of the protocol. It allows the system, under the optimistic condition that ``all processes have cast their votes'' , to achieve block commitment in just a single round of voting.

\subsection{Implementation}
The design of the pBeeGees-CB algorithm requires a few modifications to the original pBeeGees algorithm.

First, when a process casts a vote, the communication method changes from a unicast to the next view's leader to a broadcast to all processes within the system.

Second, if a correct process receives votes for a block $B_v$ from all processes ($3f+1$) and all are of the \textit{vote\_normal} type, it can commit this block immediately.

Subsequently, each process needs to modify the block ranking rule. Based on the original rule, if two blocks, $B_1$ and $B_2$, are in the same view, and their corresponding proven blocks, $B_{QC}1$ and $B_{QC}2$, are also in the same view, a further comparison is made based on the number of votes for $B_1$ and $B_2$ in the \texttt{tmo\_set}. If either block has received more than $f+1$ votes, that block is given a higher rank. Otherwise, their ranks are considered equal, and either can be chosen.

\subsection{Communication Complexity}
Since pBeeGees-CB changes the one-way sending of votes to a broadcast, the communication complexity is $O(N^2)$ in the absence of timeouts. When a timeout occurs, because the size of the broadcast block is not affected by the votes, it remains $O(N^2)$. Therefore, the total communication complexity is also maintained at $O(N^2)$.

\section{Correctness Proof}
This section proves the safety and liveness properties of the pBeeGees protocol
and its boosted commit extension, pBeeGees-CB.
\subsection{Correctness of the pBeeGees Protocol}
\subsubsection{Safety Proof}

\begin{lemma} In the same view, at most one block can collect a QC (of any type).
\end{lemma}
\begin{proof}This follows from the fact that a correct process votes at most once per view, and two QCs would require more votes than available in the system.
\end{proof}

\begin{lemma} If a block B passes validation, all of its ancestor blocks are also valid.
\end{lemma}
\begin{proof}
This follows from the protocol's validation rule, which recursively traces back to the latest block on the chain that has collected a QC.
\end{proof}

\begin{lemma}
\label{lem:equivocation_proof}
If block B collects a $QC_{!prud}$, and there exists a valid block $B'$ such that $B'.v > B.v$ and $B'$ is not a descendant of B, then $B'$ or one of its ancestors must contain an Equivocation Proof.
\end{lemma}

\begin{proof}Assume that such a $B'$ exists, and find its ancestor $B_f$ with a view greater than $B.v$. Then, by analyzing the view of $B_f$'s parent $B_p$ relative to $B$, we consider two cases: if $B_p.v < B.v$, it would violate the protocol's timeout mechanism; if $B_p.v = B.v$, it would result in a detectable equivocation. Therefore, the lemma holds.
\end{proof}

\begin{lemma} 
\label{lem:descendant_chain_1}
If a correct process commits block B (which has a $QC_{!prud}$), and $B'$ is the nearest subsequent block to collect a QC of any type, then for any block $B^*$ that collects a $QC_{!prud}$ where $B.v < B^*.v < B'.v$, $B^*$ must be a descendant of B.
\end{lemma}

\begin{proof} If $B^*$ conflicts with $B$, it also conflicts with $B'$. By \autoref{lem:equivocation_proof}, this means an equivocation proof exists between $B'$ and $B$. According to the commit rule, $B'$'s QC would be marked as \texttt{equv}, preventing $B$ from being committed, which contradicts the initial assumption.
\end{proof}

\begin{lemma} 
\label{lem:descendant_chain_2}
If a correct process commits block B (which has a $QC_{!prud}$), and $B'$ is the nearest subsequent block to collect a $QC_{!prud}$, then for any block $B^*$ that collects a QC where $B^*.v > B'.v$, $B^*$ must be a descendant of B.
\end{lemma}

\begin{proof} If $B^*$ conflicts with B, the leader of $B'$ must have equivocated. This implies an ancestor of $B^*$, say $B''$, is in the same view as $B'$. By the protocol's ranking rule, when comparing the ranks of $B'$ and $B''$, $B''$ can only have a higher rank if the block proven to by its $QC_{!prud}$, say $B_t$, have a higher rank than B, which is the block proven by $B'$'s $QC_{!prud}$. At this Point, $B.v<B_t.v<B'.v$, according to \autoref{lem:descendant_chain_1}, $B_t$ is a descendant of B. So is $B^*$.
\end{proof}

\begin{theorem}
[pBeeGees Safety Theorem]
\label{thm:pbegees_safety}
Any two committed blocks, B and $B'$, must belong to the same chain.
\end{theorem}
\begin{proof} Without loss of generality, assume $B.v < B'.v$. Let $B$ be committed because of block $B_c$, meaning both B and $B_c$ collected a $QC_{!prud}$. According to \autoref{lem:descendant_chain_1} and \autoref{lem:descendant_chain_2}, if $B'$ collects a $QC_{!prud}$, it must be a descendant of B. The theorem is thus proven.
\end{proof}

\subsubsection{Liveness Proof}
Define $F(v)$ as the waiting time for a correct process in view $v$. If the process does not enter a new view after this time, a timeout occurs. This study selects $F(v)=5\Delta$.
\begin{lemma} If a correct process p enters view v, then for every view $v' < v$, some correct process has entered $v'$.
\end{lemma}
\begin{proof}{This follows from the fact that entering a new view requires a QC or TC from the previous view, ensuring sequential progress through views.}
\end{proof}

\begin{lemma} Every correct process will continuously advance to higher views.
\end{lemma}
\begin{proof} We assume that a correct process $h$ in the lowest view never advance. Other processes would continue to advance to higher views, eventually reaching a view with a correct leader. The proposal from this leader or a timeout message from other view would eventually reach $h$, forcing it into a higher view, which is a contradiction.
\end{proof}

\begin{lemma}
\label{lem:leader_entry}
Every correct leader will eventually enter its designated view, and no later than $2\Delta$ after the first correct process arrives in that view.
\end{lemma}
\begin{proof}
Assume the correct leader $l_v$ ultimately fails to enter view $v$. At time $t$, the first correct process $h$ enters view $v$. It must have entered this view via $TC_{v-1}$, and it has collected at least $f+1$ timeout messages from correct processes.Within time $t+\Delta$, these $f+1$ timeout messages will reach all correct processes, and these correct processes will in turn broadcast their own timeout messages. Finally, within time $t+2\Delta$, all correct processes will have collected $2f+1$ timeout messages and will enter view $v$, which contradicts the assumption.
\end{proof}

\begin{lemma}
If $F(v)=5\Delta$ and the views v and v+1 are both correct views, then all correct processes will vote for the block $B_v$ proposed in view v, and all will receive the proposal from the leader of the view v + 1, $l_{v+1}$, thereby entering view v+1.
\label{lem:vote_and_advance}
\end{lemma}
\begin{proof}
According to \autoref{lem:leader_entry}, if the first correct process enters view $v$ at time $t$, its leader $l_v$ will make a proposal before $t + 2\Delta$. This proposal will reach all correct processes by $t + 3\Delta$, and it will receive $2f+1$ votes before $t + 4\Delta$. After this, $l_{v+1}$ generates $QC_v$ and proposes $B_{v+1}$. This proposal will arrive at all correct processes before $t + 5\Delta$, causing them to enter view $v+1$ and reset their timers.
\end{proof}

\begin{theorem}[pBeeGees Liveness Theorem] After GST, under ideal conditions, if there exist correct views v, v+1, $v'$, and $v'+1$ where $v' > v+1$, then the proposed block $B_v$ from view v will be committed by all correct processes within $\Delta$ after the first correct process enters view $v'+1$.
\end{theorem}

\begin{proof} By the logic of \autoref{lem:leader_entry} and \autoref{lem:vote_and_advance}, in views $v$, the leader $l_{v+1}$ will successfully collect votes and form $QC_v$. Therefore, it propose block $B_{v+1}(QC_v)$ (We use $B_{v+1}(QC_v)$ to denote that the block $B_{v+1}$ carries the quorum certificate $QC_v$). According to Safety Proof, All subsequent valid blocks will be descendants of $B_{v+1}$. Similarly, leader $l_{v'+1}$ will propose $B_{v'+1}(QC_{v'})$. When this proposal is delivered to all correct processes, they will observe that $QC_{v'}$ and $QC_v$ are on the same chain and will therefore commit $B_v$. The theorem is thus proven.
\end{proof}

\subsection{Correctness of the pBeeGees-CB Protocol}
The accelerated commitment mechanism does not affect the liveness of the algorithm; therefore, the following will only provide a safety proof for  pBeeGees-CB.
\subsubsection{Safety Proof}
\begin{lemma}
\label{lem:descendant}
For any valid block $B'$, if $B'.v > B_{boost}.v$, then $B'$ is a descendant of $B_{boost}$.
\end{lemma}

\begin{proof}
Assume for contradiction that $B'$ conflicts with $B_{boost}$ for some $B'.v > B_{boost}.v$. Let $B_f$ be the nearest ancestor of $B'$ with a view greater than $B_{boost}.v$, and let $B_p$ be its parent. By construction, $B_p.v \le B_{boost}.v$. For $B_f$ to be valid, it implies $B_p$ and $B_{boost}$ must have competed, forcing $B_p.v = B_{boost}.v$. When comparing candidates, the ranking rule must select a parent for $B_f$. Since $B_{boost}$ has votes from all correct processes, it outranks $B_p$ by vote count. Thus, any honest node would have chosen $B_{boost}$ as the parent, which contradicts that $B_p$ is the parent.
\end{proof}

\begin{lemma}
\label{lem:no_conflict}
Any two boost-commit blocks, $B_{boost}$ and $B'_{boost}$, do not conflict.
\end{lemma}

\begin{proof}
This follows directly from Lemma~\ref{lem:descendant}.
\end{proof}

\begin{lemma}
\label{lem:normal_boost_descendant}
If a block $B$ is committed normally by some process, and a block $B_{boost}$ is committed via boosting by some process, and if $B_{boost}.v > B.v$, then $B_{boost}$ is a descendant of $B$.
\end{lemma}

\begin{proof}
Let $B_c$ be the block that caused $B$ to be committed (i.e., the process committed $B$ upon receiving the proposal for $B_c$). If $B_c.v \ge B_{boost}.v$, by Lemma~\ref{lem:descendant}, it is clear that $B_c$, $B_{boost}$, and $B$ are on the same chain. If $B_c.v < B_{boost}.v$, then based on Lemmas~\ref{lem:descendant_chain_1} and ~\ref{lem:descendant_chain_2}, since $B_{boost}$ is a valid block and has received a majority of votes, it is clear that the lemma holds.
\end{proof}

\begin{theorem}[Boost Commit Safety Theorem]
\label{thm:pbegeesCB_safety}
The boost commit mechanism preserves the safety of pBeeGees.
\end{theorem}

\begin{proof}
From Lemma~\ref{lem:descendant}, Lemma~\ref{lem:no_conflict}, Lemma~\ref{lem:normal_boost_descendant}, and  Theorem~\ref{thm:pbegees_safety}, when the boost commit mechanism is configured, any two committed blocks $B$ and $B'$, regardless of whether they were committed via boosting, lie on the same chain.
\end{proof}

\section{Evaluation}

This section compares pBeeGees-CB, pBeeGees with the classic Fast-Hotstuff and Chained-Hotstuff algorithms. The BeeGees protocol was omitted from our performance comparison as its fundamental security and liveness flaws render it an unsuitable benchmark. The following sections will first introduce the experimental setup, then present the experimental results with a detailed analysis.

\subsection{Experimental Setup}

Our implementation is based on the Go language and the Relab Hotstuff project\cite{resilient_systems_lab_relabhotstuff_nodate}. The experiment simulates a global Wide Area Network (WAN) with an average latency of 250ms and introduces a 10\% probability of a 500ms severe network delay. The maximum network delay ($\Delta$) is set to 1 second. We primarily evaluate the protocol's performance across two dimensions: fault tolerance and scalability.
First, we tested the protocol's \textbf{fault tolerance} in a network of $n$=7 ($f$=2), evaluating its performance under ``stopping fault'' probabilities of 0\%, 10\%, 25\%, and 50\%. Subsequently, to test \textbf{scalability}, we increased the network to $n$=16 ($f$=5) and repeated the experiment under the 10\% fault probability condition.

\subsection{Experimental Results}
\subsubsection{Fault Tolerance Analysis}
\begin{figure}
    \centering
    \includegraphics[width=\textwidth]{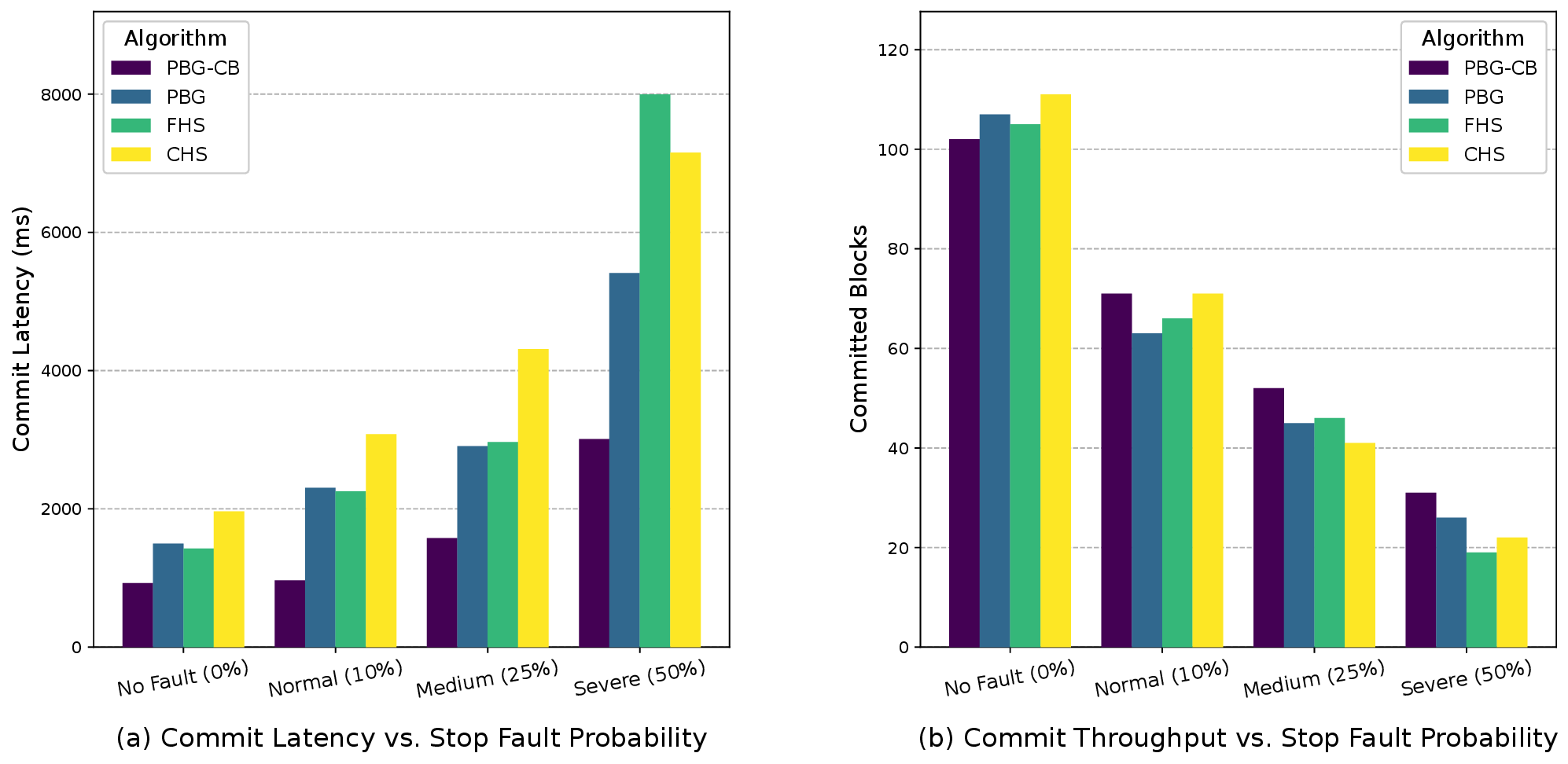}
    \caption{Performance comparison under different stop-fault probabilities with n=7, f=2. }
    \label{fig:fault_tolerance}
\end{figure}

As shown in ~\autoref{fig:fault_tolerance}a, under fault-free conditions, pBeeGees-CB exhibits the lowest commit latency, followed by pBeeGees, FHS, and CHS. This aligns with the theoretical number of message-passing rounds required to commit a block. pBeeGees-CB is the fastest (2 rounds), followed by pBeeGees and FHS (5 rounds), and CHS (7 rounds). 

Although PBG-CB needs to collect votes from all processes to commit a block, which makes PBG-CB's actual performance is slightly slower than the theoretical ratio(i.e.,2:5:5:7) suggests, but it is still significantly superior to the other algorithms, with a latency of about half that of FHS.

As the number of crash faults and the communication load increase, it becomes progressively more difficult for PBG-CB to collect votes from all processes. As a result, the gap in commit latency between it and the other algorithms gradually narrows. However, PBG-CB still maintains a clear advantage. This is because crash faults, especially when a leading process fails to propose, cause discontinuities in the blockchain's view. This prevents FHS and CHS from committing blocks, forcing them to wait for 5-7 more rounds of message passing after a timeout to potentially achieve commitment. PBG does not require a continuous view, but in such scenarios, processes still need to wait for a timeout and then validate multiple blocks, leading to a manifold increase in cryptographic overhead.

Consider a scenario where block $B_v$ receives votes from all processes in view $v$, but the leader of view $v+1$ experiences a crash fault. After a timeout, the new leader $l_{v+2}$ proposes block $B_{v+2}$. PBG-CB can commit $B_v$ as soon as it receives the votes. In contrast, PBG must wait for the timeout and for $B_{v+2}$ to form a Quorum Certificate (QC) before it can commit $B_v$. Furthermore, after validating $B_{v+2}$, processes need to validate $B_v$ again. Therefore, PBG-CB has a significant advantage in commit latency compared to all the other algorithms.

Meanwhile, In most experimental settings, PBG and FHS exhibit similar commit latencies. However, under severe crash faults (e.g., a 50\% crash probability), FHS suffers from frequent view changes, leading to significantly increased latency. In contrast, PBG, despite also face higher overhead, can still gradually reach commitment and thus outperforms FHS.
 It should be clarified that PBG does not only outperform FHS in extreme crash fault scenarios. In practical applications, the system can appropriately reduce the assumed maximum message delay value to decrease the waiting time before a timeout, aiming to improve system efficiency. In this case, if the actual message delay is greater than the reduced value, it can be treated as a crash fault or message loss. In this context, PBG demonstrates more stable performance than FHS. In other words, PBG's performance is on par with FHS in ideal network conditions, but it shows stronger adaptability in adverse network environments and in the presence of crash faults.

~\autoref{fig:fault_tolerance}b shows that throughput remains similar across all protocols. This is because once a block is committed, its ancestor blocks are also committed, resulting in a comparable number of total committed blocks.
\begin{figure}[htbp]
    \centering
    \includegraphics[width=\textwidth]{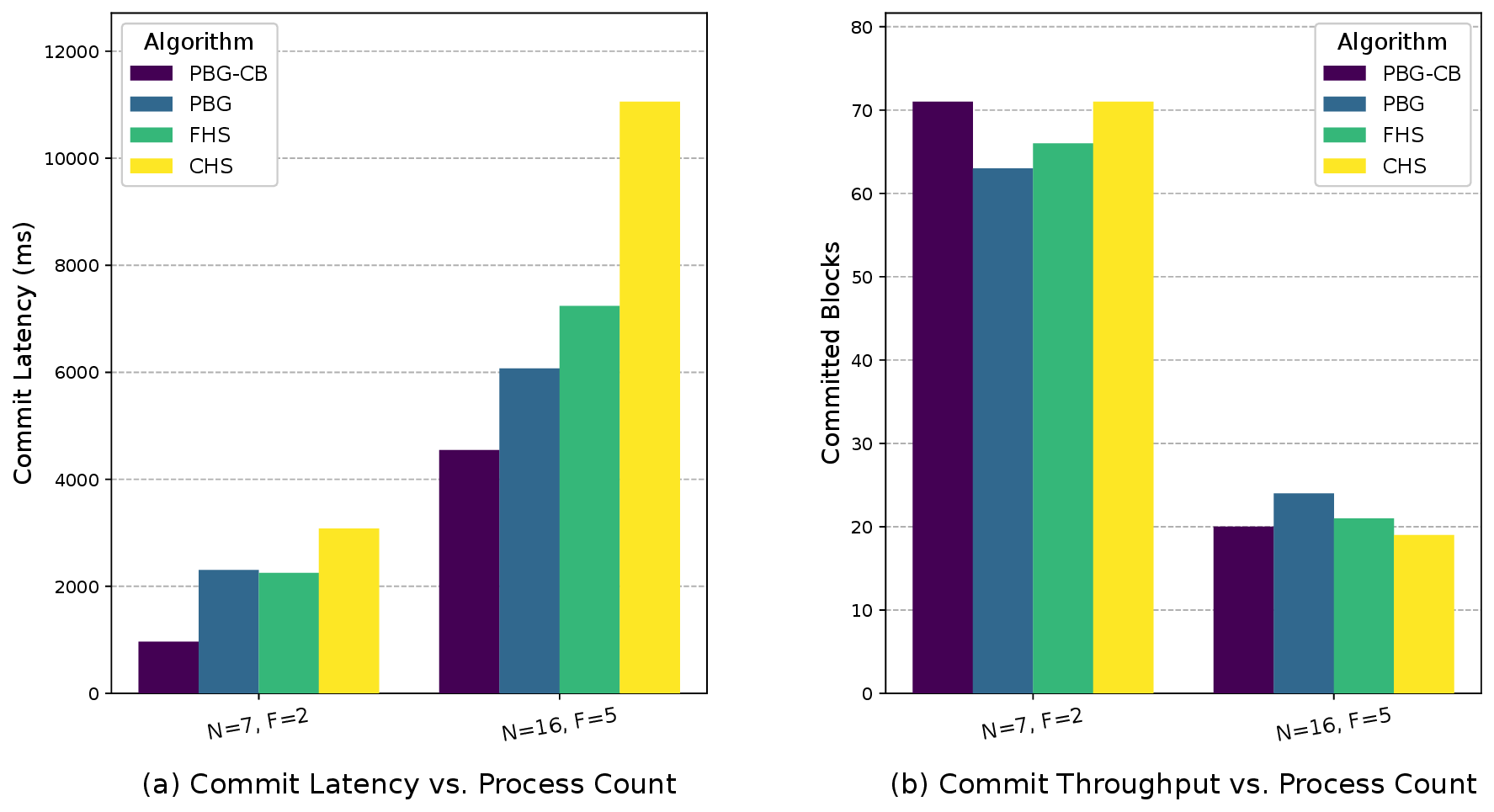}
    \caption{Performance comparison for different numbers of processes (scalability) under a normal (10\%) stop-fault probability. }
    \label{fig:scalability}
\end{figure}

\subsubsection{Scalability Analysis}

As shown in ~\autoref{fig:scalability}, increasing the number of processes degrades performance for all protocols, causing longer commit latency and lower throughput. This is primarily due to the increased communication and cryptographic overhead from collecting more votes to form QCs. This impact is particularly pronounced for pBeeGees-CB, as it requires votes from all processes, unlike other protocols that only need $2f+1$. Consequently, the performance gap between pBeeGees-CB and pBeeGees narrows as the network grows. 

Nevertheless, both pBeeGees and pBeeGees-CB consistently outperform traditional HotStuff variants like FHS and CHS in commit latency, demonstrating their significant advantages in larger-scale environments.

\section{Related Work}
The FLP impossibility theorem~\cite{fischer1985impossibility}, established by Fischer, Lynch, and Paterson, states that in a purely asynchronous network, consensus cannot be guaranteed in a finite time if even a single process crashes. This fundamental limitation led subsequent consensus algorithms to adopt one of two approaches: assuming a partially synchronous network model\cite{chanPaLaSimplePartially2018,castro1999practical, yin2019hotstuff} or employing probabilistic methods in asynchronous settings to ensure consensus is reached with arbitrarily high probability\cite{miller2016honey,duan_beat_2018,gao2022dumbo}.

The liveness of these algorithms hinges on an efficient View Synchronizer\cite{castro1999practical,civit2024byzantine,lewis2022quadratic,lewis2023fever,lewis2024lumiere,Team2021DiemBFTVS}, which ensures nodes enter and remain in the same view long enough to achieve consensus. Diem\cite{Team2021DiemBFTVS} provides a concise and rigorously provable view coordinator, therefore this study adopts the same mechanism.

%
%
%
\bibliographystyle{splncs04}
\bibliography{pBeeGees}

\end{document}